\documentclass[12pt,a4paper]{article}
\usepackage[centertags]{amsmath}
\usepackage{amsfonts,amsthm,amssymb}
\usepackage{amssymb}
\usepackage{amsmath}
\usepackage{graphicx}
\usepackage{booktabs}
\usepackage{appendix}
\usepackage[hmargin=2.6cm,vmargin=2.6cm]{geometry}
\usepackage{etoolbox}
\usepackage{tikz}
\usepackage{soul}
\usepackage{comment}
\usepackage{float}
\usepackage{setspace}
\usepackage{xcolor,framed}
\usepackage[colorinlistoftodos,prependcaption,textsize=tiny]{todonotes}
\usepackage{soul}   
\setstcolor{red}
\usepackage{natbib}

\usepackage{hyperref}
\hypersetup{colorlinks,citecolor=blue}

\usepackage{threeparttable}
\usepackage{array}
\usepackage{booktabs}

\usepackage{blindtext}

\theoremstyle{definition}
\newcommand{\real}{\mathbb{R}}

\newcommand{\p}{\mathbf{P}}

\newcommand{\tp}{\tilde{\mathbf{P}}}
\newcommand{\pr}{\mathcal{P}}

\newcommand{\q}{\mathbf{Q}}
\newcommand{\bb}{\mathbf{B}}
\newcommand{\vv}{\mathbf{V}}
\newcommand{\ww}{\mathbf{W}}

\newtheorem{lemma}{Lemma}

\newtheorem{proposition}{Proposition}

\newtheorem{theorem}{Theorem}

\newtheorem{corollary}{Corollary}

\newtheorem{definition}{Definition}
\newtheorem{example}{Example}
\newtheorem{exer            cise}[theorem]{Exercise}

\newtheorem*{lem:main}{Lemma \ref{lem:main}}

\newtheorem*{cor:int}{Corollary \ref{cor:int}}

\DeclareMathOperator{\var}{Var}
\DeclareMathOperator{\cav}{cav}
\DeclareMathOperator{\pol}{Pol}

\DeclareMathOperator{\mn}{MN}
\DeclareMathOperator{\bin}{Bin}

\begin{document}
	
	\title{A Population's Feasible Posterior Beliefs\footnote{Itai acknowledges support from the Ministry of Science and Technology
			(grant 19400214). Yakov acknowledges support from the Israel Science Foundation (grant 2021296).}}
	\author{Itai Arieli\thanks{%
			Faculty of Industrial Engineering and Management, Technion: The Israel Institute of Technology,
			iarieli@technion.ac.il}, Yakov Babichenko\thanks{%
			Faculty of Industrial Engineering and Management, Technion: The Israel Institute of Technology,
			yakovbab@technion.ac.il}}
	\maketitle
	\begin{abstract}
		We consider a population of Bayesian agents who share a common prior over some finite state space and each agent is exposed to some information about the state.
		%receives a private signal and forms a posterior belief. 
		We ask which distributions over empirical distributions of posteriors beliefs in the population are feasible. 
		We provide a necessary and sufficient condition for feasibility. We apply this result in several domains. 
		First, we study the problem of maximizing the polarization of beliefs in a population.
		Second, we provide a characterization of the feasible agent-symmetric product distributions of posteriors.
		Finally, we study an instance of a private Bayesian persuasion problem and provide a clean formula for the sender's optimal value.  
	\end{abstract}
	\section{Introduction}
	Economics is deeply concerned with the question of whether agents' behaviors can be rationalized in the sense that they can be explained as an equilibrium behavior in a Bayesian model with rational agents. One such related question is under what conditions agents’ beliefs, i.e., posterior distribution, are compatible with Bayes’s rule?  This question is well understood for the single agent where by the splitting lemma (\citep{aumann1995repeated, blackwell1951comparison}) feasibility simply means that the expectation of the posterior equals the prior. But the characterization of feasible posterior distribution for two or more agents is trickier.  
	
	The first paper to address the feasibility of posterior distributions for more than one agent is  \citep{dawid1995coherent} who characterize feasible distributions for two agents. 
	Recently a growing number of papers have examined the feasibility of posterior distributions and their connection to information design problems. \citep{arieli2021feasible} connect the notion of feasibility to the agreement theorem of \citep{aumann1976agreeing} and to the related no-trade theorem of \citep{milgrom1982information}. A follow-up paper by \citep{morris2020notrade} generalizes the result in \citep{arieli2021feasible} to an arbitrary finite state space. Despite these elegant connections, in general, determining whether a given distribution of posterior beliefs is feasible is a tough question.
	
	In this work we study an \emph{anonymous}  variant of the feasibility question. The object whose feasibility we are willing to determine is a distribution over the \emph{empirical distributions of posterior beliefs in the population} rather then a distribution over \emph{profiles of posteriors}. The difference between these two variants is that a profile of posteriors indicates which agent holds which posterior, while the anonymous version of the problem only reveals how many agents hold each posterior without revealing their identity.

	%In this work we study feasibility of posteriors distribution and their applications. For example consider an econometrician  who observes samples of empirical samples of posteriors in a population of agents. Each such empirical sample contains the fraction of the agents in the population possessing each posterior. The econometrician then eventually learns the distribution of the realized empirical samples and would like to know whether such a distribution is feasible and can be obtained from an \emph{information structure} using Bayes rule.
	%Unlike previous works, we focus on an aggregate distribution of posteriors which are simply empirical samples over posteriors. The main difference is that in our approach we do not distinguish between agents and we are only interested in the aggregative distribution of the posterior. 
	
	%Concretely, consider a population of $n$ agents and the posteriors are restricted to the set $\{1/4,3/4\}$ and the prior is $\frac{1}{2}$. Any signal realization corresponds to a certain fraction of agents that receive the signal $\frac 3 4 $ (where the complement fraction receive the signal $\frac 1 4 $). Thus we can identify the empirical posterior distributing as with a distribution over $\{\frac{j}{n}|j\in[n]\}$. We ask which such distribution is feasible?
	
	Surprisingly, we show that determining feasibility for the anonymous variant of the problem is relatively easy, especially in the case where the number of different posteriors for an individual in the population is low. 
	Our characterization (Theorem \ref{theorem:main}) provides a necessary and sufficient condition for  a distribution over empirical samples of a population to be feasible. The characterization relies on the observation that conditional on a state the frequencies of posterior beliefs in the population can be identified. We show that every combination of spreads of these expected frequencies (one for each state) is feasible.
	%for feasibility generalizes the splitting lemma in the following sense. The splitting lemma asserts that the posterior distribution of the individual is feasible if and only if it constitutes a ``split'' to the prior. In our case, when the number of states is $m$, given a distribution over empirical distributions we identify $m$ conditional empirical distributions of posteriors; one for each state. We show that the distribution is feasible if and only if it constitutes a ``split'' of the convex combination of these conditional distributions. The weights in this convex combination for each conditional distribution equals the prior probability of the corresponding state. 
	Unlike the feasibility characterization by \citep{arieli2021feasible} and \citep{morris2020notrade}, our characterization does not become more complex as the number of agents in the population increases. Furthermore, the complexity of  our characterization does not increase with the number of states either. However, the parameter of the number of possible posteriors in the population has an effect on the complexity of our characterization.
	
	%{\color{blue}To obtain our characterization we identify the ``average belief'' which is a distribution over posteriors beliefs and it is obtained by considering the distribution of the average posterior belief in the population. This distribution gives rise to $m$  conditional distributions one fro each state of the world. Our main result shows that a distribution over empirical posterior samples (global change)is feasible iff conditional on any state the distribution is a mean preserving spread or a ``split'' of the corresponding conditional distribution. } 

	As was mentioned above, the feasibility problem is an intriguing instance of rationalization in the case where an econometrician observes the distribution over empirical distributions of posteriors in the population. In the context of rationalization, it is natural to consider another variant where the econometrician observes a single sample; i.e., she observes a single empirical distribution of posteriors and wants to determine whether it is possible that the agents are Bayesian.\footnote{\citep{shmaya2016experiments,de2019rationalizing} have studied Bayesian rationalization of the behavior of a single agent in a repeated environment.} 
	%In a different related paper  study a repeated experiment where agents have to report the most likely state among finitely many state as a function of the information they received. They show that any behavior is consistent with standard updating.} 
\citep{molavi2021tests} shows that every empirical distribution of posteriors can be rationalized.
This question is equivalent to understanding which empirical distributions of posteriors might belong to the support of a feasible distribution. Our results can be viewed as an extension of \citep{molavi2021tests}; however, we go a step further and determine the entire set of feasible distributions rather than focus on points that might belong to the support.   

%\citep{molavi2021tests} who also studies feasible realization of posterior distribution for a population of agents. In distinction with the papers mentioned above \citep{molavi2021tests} studies feasibility for a single realization of posteriors, or equivalently, the beliefs that might belong to the support of a feasible distribution. Roughly speaking, \citep{molavi2021tests} shows that every belief might belong to the support. This consequence can follow also from our characterization of the feasible beliefs.

We apply our main theorem in three different domains.

\paragraph{Polarization} The question of how polarized the posteriors of two Bayesian agents might be has been previously studied. 
Different measures of polarization have been suggested, including the distance between the posteriors \citep{dubins1980maximal}, the distance to the power $\alpha$ \citep{arieli2021feasible}, and the indicator function of the posterior's distance being above some threshold \citep{burdzy2019contradictory,burdzy2020bounds}.    
The case of a quadratic distance can be naturally generalized to multiple agents by considering the variance of the belief's in the population. The maximal polarization for two agents is obtained by revealing the state to one agent and providing no information to the other agent (\citep{arieli2021feasible}). We extend this result to every population of agents of even size: the maximal polarization is obtained by revealing the state to half of the population and providing no information to the other half. Furthermore, we extend this result to an arbitrary state space (not necessarily binary).

\paragraph{Private private information} A fascinating instantiation of the feasibility problem is that of product distributions. If the distribution of belief profiles is a product distribution then the information that agent $i$ receives does not affect her belief about the posteriors of her peers. \citep{privateprivate} refer to this property as \emph{private private information} of the agents. The feasibility of product distributions has been considered in \citep{privateprivate,arieli2021feasible,gutmann} and we currently have a good understanding of this problem for the two-agent case. However, the feasibility of product distributions for more then two agents currently does not have a characterization. We contribute to this line of literature by providing a characterization of feasible \emph{agent-symmetric} product distributions for an arbitrary number of agents. This characterization takes a simple form in the case where the marginals of the product distributions are binary-supported. 

\paragraph{Private Bayesian Persuasion}
%We further apply our result to private Bayesian persuasion (see \citep{arieli2019private}). 
The Bayesian persuasion setting is quite well understood in the case where the sender communicates with the receivers publicly. In particular the value of the sender is characterized by the notion of concavification \citep{kamenica2011bayesian,aumann1995repeated}. However, if the sender is allowed to communicate privately with each receiver, the characterization of the value is known only in special cases. A particular case that has been studied in \citep{arieli2019private,barman,dughmi2019algorithmic} is the following. A marketer (sender) tries to persuade agents (receivers) to adopt a product whose quality is either high or low.
\citep{arieli2019private,barman,dughmi2019algorithmic} characterize optimal policies for the sender in special cases of sender's utility function such as supermodular or anonymous submodular utilities. We provide an additional natural instance of the problem where the optimal sender's value (and the optimal policy) can be characterized. We consider the case with a homogeneous population of agents (i.e., all agents have an identical utility function) and the sender's utility function is a monotonic anonymous function (i.e., a function of a fraction of the adopters). Interestingly, similarly to the standard Bayesian persuasion model (\citep{kamenica2011bayesian}), our characterization also relies on the notion of concavification. However, the concavification in our case is applied in a different domain.

We illustrate the key ideas of our main result via an example in Section \ref{sec:warm}. Section \ref{sec:main} proceeds with the model and the main result.  Section \ref{sec:polar} studies the implications of our results to maximal polarization of beliefs in a population. In Section \ref{sec:vector} we discuss the connections of our anonymous version of feasibility with the previously studied (non-anonymous) version and present an application to private private information. Finally, Section \ref{sec:pbp} demonstrates an application in a Bayesian persuasion framework.

\section{Warm-up Example}\label{sec:warm}

Let $\Omega=\{0,1\}$ be a binary state space with a common prior $\mu=\frac{1}{2}$. We identify $\Delta(\Omega)$ with $[0,1]$, which indicates the probability that is assigned to $\omega=1$. An information structure for $n$ agents comprises $n$ measurable sets $\{S_i\}_{i=1}^n$ as well as a probability distribution $G\in\Delta(\Omega\times S_1\times\cdots\times S_n)$. We note that almost surely any realization of $n$ signals $(s_1,\ldots,s_n)$ defines $n$ posterior beliefs $G(\omega|s_i)$ for $i\in [n]$. In the example of the present section we assume that $G(\omega|s_i)\in \{ \frac{1}{4},\frac{3}{4} \}$ almost surely for every $i\in [n]$; namely, any information that an agent is exposed to results in a posterior of either $\frac{1}{4}$ or $\frac{3}{4}$.
The information structure $G$ and the signal realization $s=(s_1,\ldots,s_n)$ induce a \emph{distribution of posteriors in the population}. Since we have assumed that there are only two possible posteriors we can identify the distribution of posteriors in the population with a fraction $H_s\in \{0,\frac{1}{n},\frac{2}{n},...,1\}$, where $H_s=\frac{k}{n}$ indicates that $k$ agents in the population have the high posterior $\frac{3}{4}$, and $n-k$ agents have the low posterior $\frac{1}{4}$. 

%We denote $\Delta_n(2)=\{0,\frac{1}{n},\frac{2}{n},...,1\}$ to capture distributions over a binary state results from $n$ agents ($n$ empirical points). 

The information structure $G$ induces a distribution over distributions of posteriors in the population $\p_G\in\Delta(\{0,\frac{1}{n},\frac{2}{n},...,1\})$. We will call such a $\p_G$ the \emph{population's empirical posteriors} for short. The population's empirical posteriors anonymize the identity of the individuals and considers only the empirical distributions of the belief distribution.

A population's empirical posteriors $\p$ is called \emph{feasible} if $\p=\p_G$ for some information structure of $n$ agents. A characterization of the feasible population's empirical posteriors is given in the following proposition. We denote by $\delta_{c}$ the Dirac measure on an element $c$.

\begin{proposition}\label{pro:ex}
	A population's empirical posteriors $\p\in\Delta(\{0,\frac{1}{n},\frac{2}{n},...,1\})$ is feasible if and only if $\p$ is a mean-preserving spread of the distribution $\frac{1}{2}\delta_{\frac{1}{4}}+\frac{1}{2}\delta_{\frac{3}{4}}$.
\end{proposition}

To demonstrate the applicability of Proposition \ref{pro:ex} consider, for instance, a population of $n=9$ agents. The uniform distribution $U(\{0,\frac{1}{9},\frac{2}{9},...,1\})$ is feasible because $U(\{0,\frac{1}{9},\frac{2}{9},...,1\})$ is a mean-preserving spread of $\frac{1}{2}\delta_{\frac{1}{4}}+\frac{1}{2}\delta_{\frac{3}{4}} = U(\{\frac{1}{4},\frac{3}{4}\})$; that is, there exists an information structure such that with probability $\frac{1}{10}$ the distribution of posteriors in the population will be $i$ agents with posterior $\frac{3}{4}$ and $9-i$ agents with posterior $\frac{1}{4}$ for every $i=0,1,...,9$.

The uniform distribution $U(\{\frac{1}{9},\frac{2}{9},...,\frac{8}{9}\})$ is infeasible because $U(\{\frac{1}{9},\frac{2}{9},...,\frac{8}{9}\})$ is not a mean-preserving spread of $\frac{1}{2}\delta_{\frac{1}{4}}+\frac{1}{2}\delta_{\frac{3}{4}} = U(\{\frac{1}{4},\frac{3}{4}\})$; that is, there is no information structure such that with probability $\frac{1}{8}$ the distribution of posteriors in the population will be $i$ agents with posterior $\frac{3}{4}$ and $9-i$ agents with posterior $\frac{1}{4}$ for every $i=1,2,...,8$.

Proposition \ref{pro:ex} is a special case of our main result, which characterizes the feasibility of an arbitrary distribution. We find it useful to present the proof of Proposition \ref{pro:ex} because it contains the main proof ideas of the general case (Theorem \ref{theorem:main}).

\begin{proof}[Proof of Proposition \ref{pro:ex}]
	
	First, we show that every feasible distribution must be a mean-preserving spread of $\frac{1}{2}\delta_{\frac{1}{4}}+\frac{1}{2}\delta_{\frac{3}{4}}$. Conditional on state $\omega=1$, every agent $i\in [n]$ has a posterior of $\frac{3}{4}$ with probability $\frac{3}{4}$. Therefore, conditional on state $\omega=1$, the expected number of agents with the high posterior $\frac{3}{4}$ is $\frac{3}{4}$. Similarly, conditional on state $\omega=0$, the expected number of agents with the high posterior $\frac{3}{4}$ is $\frac{1}{4}$. Denoting by $\p^\omega$ the population's posteriors conditional on $\omega$ we obtain a decomposition of $\p$ as $\p=\frac{1}{2}\p^0+\frac{1}{2}\p^1$ where $\mathbb{E}[\p^0]=\frac{1}{4}$ and $\mathbb{E}[\p^1]=\frac{3}{4}$. We refer to $\p^0$ as the spread of $\delta\frac{1}{4}$, we refer to $\p^1$ as the spread of $\delta\frac{3}{4}$, and we get that $\p$ is a mean-preserving spread of $\frac{1}{2}\delta_{\frac{1}{4}}+\frac{1}{2}\delta_{\frac{3}{4}}$.
	
	Conversely, let $\p=\frac{1}{2}\q^0+\frac{1}{2}\q^1$ where $\mathbb{E}[\q^0]=\frac{1}{4}$ and $\mathbb{E}[\q^1]=\frac{3}{4}$ be a mean-preserving spread of $\frac{1}{2}\delta_{\frac{1}{4}}+\frac{1}{2}\delta_{\frac{3}{4}}$. Since $\p\in \Delta(\{0,\frac{1}{n},...,1\})$ we also have $\q^\omega \in \Delta(\{0,\frac{1}{n},...,1\})$. 
	Denote $\alpha^\omega_k=\q^\omega(\{\frac{k}{n}\})$.
	We define an information structure $G$ as follows. In state $\omega$ we draw an integer $k$ with probability $\alpha^\omega_k$, then we choose a subset $A\subset [n]$ of size $|A|=k$ uniformly at random (among all possible subsets of size $k$), and finally we send a high signal $h$ to all agents in the subset $A$ and a low signal $l$ to the remaining agents. We need to show that the posterior of an agent who has observed a signal $h$ is $\frac{3}{4}$ and the posterior of an agent who has observed a signal $l$ is $\frac{1}{4}$. In state $\omega$ the signal $h$ is obtained with probability $\sum_{i=0}^n \alpha^\omega_k \frac{k}{n}=\mathbb{E}[\q^\omega]$ because $\alpha^\omega_k$ is the probability that the subset is of size $k$ and since the subset is random the probability that agent $i$ is included is $\frac{k}{n}$. Therefore, the posterior of an agent who has observed signal $h$ equals 
	$$\frac{\mathbb{E}[\q^1]}{\mathbb{E}[\q^0]+\mathbb{E}[\q^1]}=\frac{3}{4}.$$
	Similar considerations show that the posterior of an agent who has observed a signal $l$ is $\frac{1}{4}$.
\end{proof}

%More generally, if the possible posteriors are $0\leq a < b \leq 1$ and the prior is $a<\mu<b$ then the population's posteriors might be viewed as elements of $\Delta(\{0,\frac{1}{n},...,1\})$ that represent the fraction of the population with the high posterior $b$. By similar arguments to those in Proposition \ref{pro:ex} we get the following characterization.

%\begin{corollary}\label{cor:binary}
	%A population's posterior $\p\in\Delta(\{0,\frac{1}{n},\frac{2}{n},...,1\})$ is $n$-feasible for the prior $\mu$ and the posteriors $a<\mu<b$ if and only if $\p$ is a mean preserving spread of  $\frac{\mu-a}{b-a}\delta_{\frac{(\mu-a)b}{(b-a)\mu}}+\frac{b-\mu}{b-a}\delta_{\frac{(\mu-a)(1-b)}{(b-a)(1-\mu)}}$.\footnote{The proof of Corollary \ref{cor:binary} is omitted because it is similar to the proof of Proposition \ref{pro:ex} and because it is a special case of the main Theorem \ref{theorem:main}.}
	%\end{corollary}

%To interpret the expression $\frac{\mu-a}{b-a}\delta_{\frac{(\mu-a)b}{(b-a)\mu}}+\frac{b-\mu}{b-a}\delta_{\frac{(\mu-a)(1-b)}{(b-a)(1-\mu)}}$ we consider the (unique) single-agent binary-signals information structure with posteriors $a,b$. 
%The coefficients $\frac{\mu-a}{b-a}$ and $\frac{b-\mu}{b-a}$ are the \emph{unconditional} probabilities of the posteriors $a$ and $b$ (correspondingly). The Dirac points $\frac{(\mu-a)b}{(b-a)\mu}$ and $\frac{(\mu-a)(1-b)}{(b-a)(1-\mu)}$ are the \emph{conditional} probabilities of the high posterior $b$ at state $\omega=0,1$ (correspondingly). 

\section{Model and the Main Result}\label{sec:main}

It turns out that the characterization of feasibility from Section \ref{sec:warm} can be naturally extended to a much more general case. Neither the restriction to binary possible posteriors nor the restriction to a binary state space are needed. We start by introducing the general model.  

For a set $A$ we denote by $\Delta_n(A)\subset \Delta(A)$ the set of empirical samples of size $n$, i.e., $\Delta_n(A)=\{\sum_{k\in [n]} \frac{1}{n}\delta_{a^k}:a^k\in A \text{ for } k\in [n]\}$ (note that any $a\in A$ may appear in the sum more than once). 

Let $\Omega=\{1,\ldots,m\}$ 
be a finite state space. An information structure for $n$ agents comprises $n$ measurable sets $\{S_j\}_{j=1}^n$ as well as a probability distribution $G\in\Delta(\Omega\times S_1\times\cdots\times S_n)$. We denote by $\mathcal{G}_n$ the set of all information structures of $n$ agents.
The marginal distribution on $\Omega$ forms the prior; i.e., the prior $\mu\in \Delta(\Omega)$ satisfies $\mu_i=G(\omega=i)$ for any state $\omega\in\Omega$. Any realization of $n$ signals $(s_1,\ldots,s_n)$ defines almost surely $n$ posterior beliefs $G(\omega|s_j)$ for $j\in [n]$. Hence $G$ and the signal realization $s=(s_1,\ldots,s_n)$ induce a \emph{distribution of posteriors in the population} $H_s\in\Delta_n(\Delta(\Omega))$ that is given by $\sum_{j\in [n]} \frac{1}{n}\delta_{G(\omega|s_j)}$.

The information structure $G$ induces a distribution over distributions of posteriors in the population $\p_G\in\Delta(\Delta_n(\Delta(\Omega)))$, which we call the \emph{population's empirical posteriors} for short.

In the case where $S_j$ is finite for every $j$, we let $G(s)$ be the unconditional probability of a signal vector 
$s\in S_1\times\cdots\times S_n$. In this case we can write
\begin{equation}\label{eq:pg}
	\p_G=\sum_{s\in S_1\times\cdots\times S_n}G(s)\delta_{H_s},
\end{equation}
where $\delta_{H_s}\in\Delta_n(\Delta(\Omega))$ is a Dirac measure on $H_s$.

\begin{definition}
	A distribution $\p\in\Delta(\Delta_n(\Delta(\Omega)))$ is called \emph{feasible} for a prior $\mu\in\Delta(\Omega)$ if $\p=\p_G$ for some information structure $G\in \mathcal{G}_n$. We denote by $\mathcal{P}_n(\mu)\subseteq \Delta(\Delta_n(\Delta(\Omega)))$ the set of all feasible distributions for a prior $\mu\in\Delta(\Omega)$ and $n$ agents.
\end{definition}

%We let $P$ be union all $n$-feasible distributions across all $n\in\Nat$ and let $\pr$ be its closure under the weak* topology. We call a distribution $\p$ \emph{feasible} if $\p\in\pr$. For $\mu\in[0,1]$ let $\p(\mu)$ be the set of all feasible distribution with a prior $\mu$.

Note that if the signal sets  $\{S_j\}_{j=1}^n$ are finite, $\p_G$ assigns probability one to measures $H_s\in\Delta_n(\Delta(\Omega))$ that are supported on some finite set $\{y_1,\ldots,y_m\}\subseteq \Delta(\Omega)$ (where $m\leq \sum_j |S_j|$). In general, we say that $\p\in\mathcal{P}$ is \emph{finitely supported} if there exists a finite set $Y\subseteq\Delta(\Omega)$ of posteriors such that $\p$ assigns probability one to the set of measures $\Delta(\Delta_n(Y))$. 
%Let $\pr_F\subset \mathcal{P}$ be the set of all finitely supported distributions. 

%Our main goal is to characterize $\pr_c$ and $\pr$.
\subsection{Main Result}
For any $\p\in\Delta(\Delta_n(\Delta(\Omega)))$ we let 
$\tp\in\Delta(\Delta(\Omega))$ be the \emph{expected} measure that $\p$ induces on $\Delta(\Omega)$. That is, for every Borel subset $B\subseteq\Delta(\Omega)$,
$$\tp(B)=\int_{\Delta(\Omega)}H(B) \mathrm d\p(H).$$
In particular, if $\p$ is finitely supported we can write $\p=\sum_{j=1}^k a_j H_j$ (where $H_j\in\Delta_n(\Delta(\Omega))$, $a_j\geq 0$ and $\sum_j a_j=1$); then 
$$\tp(B)=\sum_{j=1}^k a_jH_j(B).$$
%Let $\mu\in \Delta(\Omega)$ be the expectation of $\tp$. That is, for any state $\omega_i \in\Omega$
%$\mu_i=\int_{[0,1]}x_i \mathrm d\tp(x)$. Note that  $\p$ is feasible for some prior if and only if $\p\in\mathcal{P}_n(\mu)$. 
We define $m$ other distributions $\tp^\omega \in \Delta(\Delta(\Omega))$ for $\omega\in\Omega$ as follows:
$$\mathrm d\tp^\omega(x)=\frac{1}{\mu_\omega} x_\omega \mathrm d\tp(x).$$
If $\p$ is finitely supported we simply have $\tp^\omega(x)=\frac{x_\omega \tp(x)}{\mu_\omega}$.

Finally, let $\p'\in\Delta((\Delta(\Delta(\Omega)))$
be the the measure that is defined as follows:
$$\p'=\sum_{\omega \in[m]}\mu_\omega \delta_{\tp^\omega},$$
where $\delta_{\tp^\omega}\in\Delta((\Delta(\Delta(\Omega)))$ are Dirac measures. 
%Thus, any $\p\in\Delta(\Delta_n(\Delta(\Omega)))$ gives rise to another measure $\p'\in\Delta(\Delta_n(\Delta(\Omega)))$. 
Our main result states the following:
\begin{theorem}\label{theorem:main}
	A distribution $\p\in\Delta(\Delta_n(\Delta(\Omega)))$ is feasible for $n$ agents and a prior $
	\mu\in\Delta(\Omega)$ if and only if $\p$ is a mean-preserving spread of $\p'$.
\end{theorem}

The proof follows a similar logic to the proof of Proposition \ref{pro:ex}; however, the general case differs from the binary-posterior case in the following respect. A key observation in the proof of Theorem \ref{theorem:main}, as well as the proof of Proposition \ref{pro:ex}, is that the expected distribution of posteriors conditional on a state $\omega$ can be identified. In the binary-posterior case, conditional on any state $\omega$, \emph{every individual} in the population has on average the high posterior ($\frac{3}{4}$ in Proposition \ref{pro:ex}) with a fixed probability that is determined by the two possible posteriors. In the case of more than two posteriors we cannot deduce such a strong statement about each individual in the population.\footnote{Consider, for instance, a binary-state space with three possible posteriors $0<a<b<\mu<c<0$. Agent 1 has only two possible posteriors $a$ and $c$. Agent 2 has only two possible posteriors $b$ and $c$. Obviously, agents 1 and 2 will have different distributions of posteriors conditional on any state.} Nevertheless, it turns out that the \emph{average} distribution of posteriors over the population conditional on a state can be identified. To do so, we simply do the exercise of deducing the average \emph{conditional} ($\tp^\omega$) distribution of posteriors from the average \emph{ex-ante} distribution of posteriors ($\tilde{\p}$). The validity of the above deduction follows from the linearity of expectation and the linearity of the operator that maps an ex-ante distribution of posteriors to a conditional distribution of posteriors. Once we show that conditional on a state $\omega \in \Omega$ the average distribution of posteriors must be $\tilde{\p}^\omega$, the necessity for the mean-preserving spread (of $\p'$) condition immediately follows. To prove sufficiency we construct an information structure with the population's empirical posteriors $\p$ using similar arguments to those of Proposition \ref{pro:ex}.

%Another way to view our characterization is the following. The distribution $\tp\in\Delta(\Delta(\Omega))$ defined above may be considered as the average distribution of posteriors in the population. 
%Following Bayes rule, this measure gives rise to $m$ measures $\{\p_i\}_{i=1}^m\subseteq \Delta(\Delta(\Omega))$ which represents the average distribution of posterior in the population conditional on state $i$. It makes sense that a necessary condition for $\p$ to be feasible is that the distribution over the empirical distribution of posteriors conditional on state $i$ is a spread of $P_i$ and therefore $\p$ must be a spread of $\p'$ (the formal argument is slightly more delicate). Our characterization shows that the converse is also true. Namely, that if $\p$ is a spread of $\p'$, then $\p$ is feasible.  

\subsection{The binary-state binary-posterior case}
In the case where the state space is binary and the posteriors get only two values $a,b$ where $0\leq a < \mu < b \leq 1$, the distribution $\p$ is an element of $\Delta(\Delta_n(\{a,b\}))$ or simply an element of $\Delta(\{0,\frac{1}{n},...,1\})$ that indicates the weight of $b$ in the distribution. In the latter representation of the population's empirical posteriors the terms $\tp^\omega$ for $\omega=0,1$ represent the expected fraction of agents that have the high posterior conditional on a state. This expectation equals the probability of an agent having the high posterior in the \emph{single-agent} case. Simple calculations show that $\tp_1=\frac{(\mu-a)b}{(b-a)\mu}$ and $\tp_0=\frac{(\mu-a)(1-b)}{(b-a)(1-\mu)}$.
%\begin{itemize}
	%   \item The term $\mu_{\p}$ represents the expected  - in the population - (unconditional) probability of receiving the high signal. This expectation equals the probability of receiving the high signal in the \emph{single agent case}. Simple calculations show that $\mu_{\p}=\frac{\mu-a}{b-a}$.
	%  \item The terms $\tp_\omega$ for $\omega=0,1$ represent the expected number of agents' fraction that have the high posterior conditional on a state. This expectation equals the probability of having the high posterior in the \emph{single agent case}. Simple calculations show that $\tp_1=\frac{(\mu-a)b}{(b-a)\mu}$ and $\tp_0=\frac{(\mu-a)(1-b)}{(b-a)(1-\mu)}$.
	%\end{itemize}

Thus, we deduce the following corollary:

\begin{corollary}\label{cor:binary}
	A population's empirical posteriors $\p\in\Delta(\{0,\frac{1}{n},\frac{2}{n},...,1\})$ is feasible for the prior $\mu$ and the posteriors $a<\mu<b$ if and only if $\p$ is a mean-preserving spread of  $$\mu\delta_{\frac{(\mu-a)b}{(b-a)\mu}}+(1-\mu)\delta_{\frac{(\mu-a)(1-b)}{(b-a)(1-\mu)}}.$$
\end{corollary}

The fact that $\mu_P$ and $\tp^\omega$ can be deduced directly from the values of the posteriors $a$ and $b$ is a special feature of the binary-posterior case.\footnote{Those properties hold for the case of binary posteriors and an arbitrary state space.} In the general case, the terms $\mu_P$ and $\tp^\omega$ are not determined by the support of the posteriors but rather are determined by the distribution $\p$ as described above in the definitions of these terms. 

We also notice that for single-dimensional distributions it is easy to determine whether the distribution is a mean-preserving spread of a given binary-support distribution as follows. Let $\p\in\Delta(\{0,\frac{1}{n},\frac{2}{n},...,1\})$ be a distribution that assigns a probability of $p_i$ to $\frac{i}{n}$ and let $\alpha \delta_{a} + (1-\alpha) \delta_b$ be a binary-support distribution with $a<b$. We consider the \emph{$\alpha$-quantile distribution of $\p$} that is defined as follows. We let $k$ be the index such that $\sum_{i<k} p_i \leq \alpha < \sum_{i\leq k} p_i$. We let $q_k=\alpha - \sum_{i<k} p_i$. The $\alpha$-quantile distribution of $\p$ is denoted by $\p|_{\alpha}$ and is defined to be the distribution that assigns a probability of $\frac{p_i}{\alpha}$ to $\frac{i}{n}$ for every $i<k$ and assigns a probability of $\frac{q_k}{\alpha}$ to $\frac{k}{n}$. Simply speaking, $\p|_{\alpha}$ is the distribution of $\p$ conditional on the event that the realization belongs to the lower $\alpha$-quantile of $\p$. The following lemma characterizes the mean-preserving spread property.

\begin{lemma}\label{lem:mps}
	A distribution $\p\in\Delta(\{0,\frac{1}{n},\frac{2}{n},...,1\})$ is a mean-preserving spread of $\bb=\alpha \delta_{a} + (1-\alpha) \delta_b$ if and only if  $\mathbb{E}[\p]=\mathbb{E}[\bb]$ and $\mathbb{E}[\p|_{\alpha}]\leq a$.
\end{lemma}

\begin{proof}
	We denote by $F_\p,F_\bb$ the CDFs for $\p$ and $\bb$. We have that $\p$ is a mean preserving spread of $\bb$ if and only if $\mathbb{E}[\p]=\mathbb{E}[\bb]$ and for every $t\in [0,1]$ we have $\int_0^t F_\p(x) dx \geq \int_0^t F_\bb(x) dx $. Moreover, it is well known that it suffices to check this inequality only for points $t$ such that $F_\p(t)$ crosses $F_\bb (t)$.
	In our case $F_\bb$ is piecewise constant with values $0,\alpha,1$ in the corresponding segments $[0,a),[a,b),[b,1]$. Therefore, in our case it is sufficient to check the inequality at the single point $t=\frac{k}{n}$.
	
	Note that the CDF of $\p|_{\alpha}$ is $F_{\p|_\alpha}(x)=\alpha^{-1} F_\p(x)$. Hence we have 
	$$\int_0^\frac{k}{n} F_\p(x) dx =\alpha \int_0^\frac{k}{n} F_{\p|_\alpha}(x) dx=\alpha \left(\frac{k}{n}-\int_0^\frac{k}{n} (1-F_{\p|_\alpha}(x)) dx \right)=\alpha \left(\frac{k}{n}-\mathbb{E}[\p|_{\alpha}] \right).$$
	We also have $\int_0^\frac{k}{n} F_\bb(x) dx = \alpha (\frac{k}{n}-a)$ and therefore $\int_0^\frac{k}{n} F_\p(x) dx \geq \int_0^\frac{k}{n} F_\bb(x) dx$ if and only if $\mathbb{E}[\p|_{\alpha}]\leq a$.
\end{proof}

\subsection{Proof of Theorem \ref{theorem:main}}\label{sec:proof}

Let $\p\in\Delta(\Delta_n(\Delta(\Omega))$ be a mean-preserving spread of $\p'$. Namely, there exist $m$ distributions  $\q^\omega\in\Delta(\Delta(\Omega))$ with expectation $\tp^\omega$  such that 
$$\p=\sum_{\omega \in \Omega}\mu_\omega\q^\omega.$$

%$\p=\sum_{j=1}^k a_{j,1}\delta_{H_j}$ and

%$\p_0=\sum_{j=1}^k a_{j,0}\delta_{H_j}$
%where $a_j=\mu_\p a_{j,1}(1-\mu_\p)a_{j,0}$ for every $j\in[k]$.

We now define a signaling structure $G$ that implements $\p$ as follows.
Let $S_j=\Delta(\Omega)$ for every agent $j$ and let $G(\omega)=\mu_i$ for every state $i\in\Omega$. 
We define $G(s|\omega)$ by coupling this measure with $\q^\omega$ as follows. We draw $H\in\Delta_n(\Delta(\Omega))$ according to $\q^\omega$. Every such $H$ has the form $H=\sum_{k\in[n]}\frac{1}{n}\delta_{x^k}$ where $x^k\in\Delta(\Omega)$.
For every realized $H$ we draw a permutation $\pi:[n]\to[n]$ uniformly at random across all permutations and send signal $x^{\pi(j)}$ to every agent $j$. Thus, overall the probability that agent $j$ observes the signal $x^k$ is $\frac{1}{n}$ for every $j,k\in[n]$. 
%we generate a random partition of the $n$ agents into $|Y|$ subsets where we correspond to each $y$ a subset such that the following conditions hold: (i) in each subset $y\in Y$ there are exactly $l(y)$ agents $(ii)$ the probability of agent $i\in[n]$ to be in subset $y\in Y$ is exactly $\frac{l(y)}{n}$. We then assign to all agents $i$ in subset $y\in Y$ the signal $s_i=y$. Thus the overall probability of an agent $i$ to receive a signal $y$ given $H_j$ is $\frac{l(y)}{n}=H_j(\{y\})$.
%We define $G(s|\omega=0)$ symmetrically with respect to the distribution $\p_0$. 

We claim that $\p_G=\p$. To see this, by the construction, it is enough to show that $G(\omega|s_j=x)=x_\omega$ for every agent $j\in[n]$ and state $\omega\in\Omega$.
We note first that since $H$ is drawn according to $\q^\omega$ and since the expectation of $\q^\omega$ is $\tp^\omega$ we have that $G(s_j\in B|\omega)=\tp^\omega(B)$ for every Borel measurable set $B\subset\Delta(\Omega)$.  
Therefore, by Bayes's rule, for every $\omega\in\Omega,$ we have that
\begin{align}\label{eq:bayesr}
	G(\omega|s_j=x)=\frac{\mu_\omega\frac{\mathrm{d}\tp^\omega}{\mathrm{d}\tp}(x)}{\sum_{\omega'\in\Omega}\mu_{\omega'} \frac{\mathrm{d}\tp^{\omega'}}{\mathrm{d}\tp}(x)}  
\end{align}
with probability 1.

We recall that
$\mathrm{d}\tp^\omega(x)=\frac{x_\omega}{\mu_\omega}\mathrm{d}\tp(x).$
Therefore, $\frac{\mu_{\omega'}\mathrm{d}\tp_{\omega'}(x)}{\mathrm d \tp(x)}=x_{\omega'}$ for every state $\omega'\in\Omega$. Hence, Equation \eqref{eq:bayesr} implies that
$G(\omega|s_j=x)=\frac{x_\omega}{\sum_{\omega'\in \Omega }x_{\omega'}}=x_\omega$ as desired.

We now turn to the converse direction; namely, we show that every feasible $\p\in \pr_n$ is a mean-preserving spread of $\p'$.
For clarity of exposition we prove this statement only for the case where $\p$ is finitely supported. The extension of the proof to arbitrary distributions $\p$ is relegated to the Appendix \ref{app:trm-proof} and essentially follows from the fact that any distribution $\p$ can be approximated (in the weak* topology) by finitely supported ones.  

We assume that $\p=\p_G$ is induced by an information structure $G\in\Delta(\Omega\times S_1\times\cdots\times S_n)$ with finite signal sets $S_i$ for every $i\in[n]$ (the finite support assumption). 
In addition to the distribution $\p$, we next define a few other auxiliary measures. First, for every state $\omega\in\Omega$ let $\p^\omega_{G}\in\Delta(\Delta_n(\Delta(\Omega)))$ be the conditional distribution of $H_s$ conditional on state $\omega$. That is, similar to Equation \eqref{eq:pg},
$$\p^\omega_{G}=\sum_{s\in S_1\times\cdots\times S_n}G(s|\omega)\delta_{H_s}.$$
It follows that 
\begin{equation}\label{eq:definitionP}
	\p=\sum_{\omega\in \Omega} \mu_\omega \p^\omega_{G}.    
\end{equation}
For every $\omega\in\Omega$, let $\tp^\omega_{G}$ be the expectation of $\p^\omega_{G}$. By definition, the measure $\p^\omega_{G}$ is a mean-preserving spread of $\delta_{\tp^\omega_{G}}$. 
Therefore, in particular, by Equation \eqref{eq:definitionP} $\p$ is a mean-preserving spread of the measure 
$\sum_{\omega \in \Omega} \mu_\omega \delta_{\tp^\omega_{G}}.$ 
In order to complete the proof we need to show that $\tp_{\omega}$ as defined above equals $\tp^\omega_{G}$ for every state $\omega \in\Omega$.

For each signal $s_j\in S_j$ let $G(\ |s_j)\in\Delta(\Omega)$ be the conditional probability on $\Omega$ given the signal $s_j$. Let $H_{s_j}=\delta_{G(\ |s_j)}.$ Using this notation, for any vector of realized signals $s=(s_1,\ldots,s_n)$, one can write the measure $H_{s}$ as follows:
\begin{equation*}
	H_s=\frac{1}{n}\sum_{j=1}^n H_{s_j}.
\end{equation*}
Therefore, we have 
$$\p=\sum_{s\in S_1\times\cdots\times S_n}G(s)\frac{1}{n}\sum_{j=1}^n H_{s_j}.$$
Hence, by definition of $\tp$, for every $x\in \Delta(\Omega)$ we have 
\begin{equation*}
	\tp(\{x\})=\sum_{s\in S_1\times\cdots\times S_n}G(s)\frac{1}{n}\sum_{j=1}^n H_{s_j}(x)=
	\frac{1}{n}\sum_{j=1}^n\sum_{s_j\in S_j,G(\ |s_j)=x}G(s_j).      
\end{equation*}
Similarly, for any state $\omega\in\Omega$,
\begin{equation*}
	\tp^\omega_{G}(\{x\})=\sum_{s\in S_1\times\cdots\times S_n}G(s|\omega)\frac{1}{n}\sum_{j=1}^n H_{s_j}(x)=
	\frac{1}{n}\sum_{j=1}^n\sum_{s_j\in S_j,G(\ |s_j)=x}G(s_j|\omega).      
\end{equation*}
It follows from Bayes's rule that if $G(\omega|s_j)=x_\omega$, then $G(s_j|\omega)=x_\omega \frac{G(s_j)}{\mu_\omega}.$
Hence $\tp^\omega_{G}(\{x\})=\frac{x_\omega}{\mu_\omega}\tp(\{x\})$ for every $x\in\Delta(\Omega)$. 
Thus $d\tp_{G,\omega}(\{x\})=\frac{x_\omega}{\mu_\omega}d\tp(\{x\}).$ Therefore, for every $\omega \in \Omega$ it holds that $\tp^\omega_{G}=\tp^{\omega}$, as desired.

\section{Polarization}\label{sec:polar}
The question of how polarized the posteriors of Bayesian agents might be has been studied  in \citep{dubins1980maximal}, \citep{burdzy2020bounds}, and \citep{arieli2021feasible}. These papers focuse on the case of two receivers. In \citep{dubins1980maximal} polarization is measured by the absolute distance $|x_1-x_2|$ and more generally in \citep{arieli2021feasible}  polarization is measured by the function $|x_1-x_2|^\alpha$ for some $\alpha>0$, where $x_i\in [0,1]$ denotes the posteriors of agent $i=1,2$. The case of $\alpha=2$ is a particularly convenient case to study from a technical perspective. It is also a natural case to study for the following descriptive reason: it has a natural extension to an arbitrary number of agents as we show below.

Given a distribution of posteriors in the population $H\in \Delta_n([0,1])$, we can measure polarization via $\var(H)$. Note that in the special case of $n=2$ with posteriors $x_1,x_2$ the variance is $\frac{(x_1-x_2)^2}{4}$, which (up to a factor of 4) is exactly the case of $\alpha=2$. 

The results of \citep{dubins1980maximal} for $\alpha=1$ and of \citep{arieli2020identifiable} for $\alpha\leq 2$ show that for the case of two agents the information structure that maximizes polarization is the one that reveals the state to one agent and provides no information to the other agent.\footnote{It remains an open problem to specify the information structure that maximizes polarization for $\alpha>2$, even for the case of two agents.} This information structure achieves an expected variance of $\frac{\mu (1-\mu)}{4}$ for information structures with prior $\mu\in [0,1]$. We show that the same result is true for every even number of players: the information structure that maximizes polarization is the one that reveals the state to half of the agents and provides no information to the other half of the agents. 

The polarization of $G$ is the expected variance of the distribution of posteriors in the population. That is, for any $G\in\mathcal{G}$ let 
$V_G=\int_{\Delta([0,1])}\var(H)\mathrm dG(H)$.

\begin{proposition}\label{pro:pol}
	Let $\mu\in [0,1]$ be the prior. 
	For every even $n$, the information structure that maximizes the polarization of $n$ agents is the one that reveals the state to half of the agents and reveals no information to the other half. The expected polarization equals $\frac{\mu (1-\mu)}{4}$. 
	
	For every odd $n$, the maximal polarization is bounded between $$\left(1-\frac{1}{n^2}\right)\frac{\mu (1-\mu)}{4} \leq \max_{G\in \mathcal{G}_n} V_{G} \leq \frac{\mu (1-\mu)}{4}.$$
\end{proposition}

For an odd number of agents, Proposition \ref{pro:pol} provides asymptotically tight bounds when $n\to \infty$. It remains an open problem to provide the precise characterization of polarization for an odd number of agents. 
Interestingly, the idea of sending full information or no information to each one of the agents is no longer optimal for an odd number of agents as demonstrated in the following example.
%the bound of lower bound of Proposition \ref{pro:pol} is not tight as demonstrated in the following example.

\begin{example}
	For $n=3$ and prior $\mu=\frac{1}{2}$, consider the information structure where the state is revealed to agent $3$, and agents $i=1,2$ receive the signals $s^i_1,s^i_2$ according to the following distribution conditional on a state:
	\begin{table}[H]
		\begin{center}
			\begin{tabular}{ccccccc}
				\multicolumn{3}{c}{$\omega=0$}                                                     & \ \ & \multicolumn{3}{c}{$\omega=1$}                                                     \\
				& $s_1^2$                  & $s_2^2$                  &                                 &                              & $s_1^2$                  & $s_2^1$                  \\ \cline{2-3} \cline{6-7} 
				\multicolumn{1}{c|}{$s_1^1$} & \multicolumn{1}{c|}{1/3} & \multicolumn{1}{c|}{1/3} &                                 & \multicolumn{1}{c|}{$s_1^1$} & \multicolumn{1}{c|}{0}   & \multicolumn{1}{c|}{1/3} \\ \cline{2-3} \cline{6-7} 
				\multicolumn{1}{c|}{$s_2^1$} & \multicolumn{1}{c|}{1/3} & \multicolumn{1}{c|}{0}   &                                 & \multicolumn{1}{c|}{$s_2^1$} & \multicolumn{1}{c|}{1/3} & \multicolumn{1}{c|}{1/3} \\ \cline{2-3} \cline{6-7} 
			\end{tabular}
		\end{center}
	\end{table}
	Agent $i$'s posterior after observing $s^i_1$ is $1/3$; her posterior after observing $s^i_2$ is $2/3$.
	The resulting distribution over posterior profiles is uniform over the profiles: $$\{(\frac{1}{3},\frac{1}{3},0),(\frac{1}{3},\frac{2}{3},0),(\frac{2}{3},\frac{1}{3},0),(\frac{1}{3},\frac{2}{3},1),(\frac{2}{3},\frac{1}{3},1),(\frac{2}{3},\frac{2}{3},1)\}.$$ 
	The expected variance is $\frac{2}{6}\frac{2}{81}+\frac{4}{6}\frac{2}{27}=\frac{14}{243}$, which is above $\frac{1}{18}$, where $\frac{1}{18}$ is the maximal expected variance that can be obtained by sending to every player either  full information or no information about the state.
\end{example}

\begin{proof}[Proof of Proposition \ref{pro:pol}]
	First we show that $V_G\leq \frac{\mu (1-\mu)}{4}$ for every number of agents $n$ (odd or even). This statement can be equivalently stated with respect to a feasible population's empirical posteriors. We denote $V_\p=\int_{\Delta([0,1])}\var(H)\mathrm d\p(H)$ and we trivially have $\max_{G\in \mathcal{G}_n} V_G=\max_{\p \in \mathcal{P}_n(\mu)} V_{\p} $.
	
	Let $\p\in\mathcal{P}_n(\mu).$ By Theorem \ref{theorem:main} we can decompose $\p=\mu_\p\p_1+(1-\mu_\p)\p_0$ where $\mathbb{E}[\p^\omega]=\tp^\omega$.
	Therefore, $V_\p=\mu_\p V_{\p_1}+(1-\mu_\p) V_{\p_0}.$
	Let $\p'=\mu_\p\delta_{\tp_1}+(1-\mu_\p)\delta_{\tp_0}$. Note that since $\p\in\mathcal{P}(\mu)$ it also holds that $\p'\in\mathcal{P}(\mu).$ We claim that $V_{\p}\leq V_{\p'}$. To see this, note that for any $\omega=0,1$,
	\begin{align}
		&V_{\p^\omega}=\int_{\Delta([0,1])}\var(H)\mathrm d\p^\omega(H)=\int_{\Delta([0,1])}\mathbb{E}_H[x^2]-(\mathbb{E}_H[x])^2\mathrm d\p^\omega(H) =\label{eq:varaince}\\
		\notag&\int_{\Delta([0,1])}\mathbb{E}_H[x^2]\mathrm d \p^\omega(H)-\int_{\Delta([0,1])}(\mathbb{E}_H[x])^2\mathrm  d\p^\omega(H)=\\
		&\mathbb{E}_{\tp^\omega}[x^2]-\int_{\Delta([0,1])}(\mathbb{E}_H[x])^2\mathrm d \p(H)\leq\label{eq:expectation}\\
		&\mathbb{E}_{\tp^\omega}[x^2]-(\mathbb{E}_{\tp^\omega}[x]))^2=\var(\delta_{\tp^\omega})\label{eq:jensen}
	\end{align}
	Note that equation \eqref{eq:varaince} follows from $\var(H)=\mathbb{E}_H[x^2]-(\mathbb{E}_H[x])^2$. Equation \eqref{eq:expectation} follows from the fact that the operation $\mathbb{E}_H(x^2):\Delta(\Delta([0,1]))\to\mathbb R$ is a linear operator and so for every $\q\in\Delta(\Delta([0,1]))$ we have that $\int_{\Delta([0,1])}\mathbb{E}_H[x^2]\mathrm d \q(H)=\mathbb{E}_{\tilde\q}[x^2]$. The inequality in \eqref{eq:jensen}  follows from Jensen's inequality.
	Since $V_{\p'}=\mu_\p V_{\delta_{\tp_1}}+(1-\mu_\p) V_{\delta_{\p_0}}$ it follows that $V_{\p'}\geq V_{\p}.$ 
	
	%Therefore, given a prior $\mu$, our optimization problem is to maximize $\mu Var(\tp_1)+(1-\mu) Var(\tp_0)$ subject to the constraint 
	
	Note that $\tp_1,\tp_0$ satisfy
	$$\mathrm d \tp_1=\frac{x}{\mu}\mathrm d \tp\text{ and }\mathrm d \tp_0=\frac{1-x}{1-\mu}d \tp,$$
	for some $\tp\in\Delta([0,1])$ with expectation $\mu$, and hence the variances can be written as 
	
	\begin{align*}
		\var(\tp_1)&=\int_{[0,1]}\frac{x^3}{\mu}\mathrm d \tp (x)-\left(\int_{[0,1]}\frac{x^2}{\mu}\mathrm d \tp (x)\right)^2 \\
		\var(\tp_0)&=\int_{[0,1]}\frac{x^2(1-x)}{1-\mu}\mathrm d \tp (x)-\left(\int_{[0,1]}\frac{x(1-x)}{1-\mu}\mathrm d \tp (x)\right)^2
	\end{align*}
	
	Therefore,
	\begin{align}
		&V_{\p} \leq \notag\mu \var(\tp_1)+(1-\mu) \var(\tp_0)=\\
		&\notag \int_{[0,1]}x^3+x^2(1-x)\mathrm d \tp (x)-\frac{1}{\mu}\left(\int_{[0,1]}x^2\mathrm d \tp (x)\right)^2-\frac{1}{1-\mu}\left(\int_{[0,1]}x(1-x)\mathrm d \tp (x)\right)^2  =\\
		&\notag \int_{[0,1]}x^2\mathrm d \tp (x)-\frac{1}{\mu}\left(\int_{[0,1]}x^2\mathrm d \tp (x)\right)^2-\frac{1}{1-\mu}\left(\mu-\int_{[0,1]}x^2\mathrm d \tp (x)\right)^2=\\
		&\notag (\frac{1+\mu}{1-\mu})\int_{[0,1]}x^2\mathrm d \tp (x)-\frac{1}{\mu(1-\mu)}\left(\int_{[0,1]}x^2\mathrm d \tp (x)\right)^2-\frac{\mu^2}{1-\mu}.
	\end{align}
	Denote $A=\int_{[0,1]}x^2\mathrm d \tp (x)$ and we obtain $V_{\p}\leq (\frac{1+\mu}{1-\mu})A-\frac{1}{\mu(1-\mu)}A^2-\frac{\mu^2}{1-\mu}$.
	This is a degree 2 polynomial whose global maximum is 
	$\frac{\mu(1-\mu)}{4}$ that is obtained at $A=\frac{\mu(1+\mu)}{2}$. 
	
	To show that $\frac{\mu (1-\mu)}{4}$ is achievable for an even number of players, we calculate the expected variance for the information structure that reveals the state to $\frac{n}{2}$ of the players and reveals nothing to the remaining $\frac{n}{2}$ players and get precisely $\frac{\mu(1-\mu)}{4}$.
	
	To show that $(1-\frac{1}{n^2}) \frac{\mu (1-\mu)}{4}$ is achievable for an odd number of players, we calculate the expected variance for the information structure that reveals the state to $\frac{n+1}{2}$ of the players and reveals no information to the remaining $\frac{n-1}{2}$ players and get the desired expression.
\end{proof}

\subsection{Multiple States}

For an arbitrary finite state space $\Omega$, one might measure polarization of a distribution of posteriors in the population via $\pol(H)=\mathbb{E}_{x\sim H} [\|x-\mathbb{E}[H]\|_2^2]$, where $x$ and $\mathbb{E}[H]$ are $|\Omega|$-dimensional vectors of beliefs in the simplex. Note that $\pol(H)=\sum_{\omega\in \Omega} \var(X_\omega)$, where $X_\omega$ denotes the random variable of the posterior probability that is assigned to state $\omega\in \Omega$. Therefore, for an arbitrary state space we obtain the following corollary.

\begin{corollary}
	For every prior $\mu\in \Delta(\Omega)$ and every even $n$ the information structure that maximizes the expected polarization 
	(i.e., $\mathbb{E}_{H\sim P} \pol(H)$) of $n$ agents is the one that reveals the state to half of the agents and reveals no information to the other half.
\end{corollary}

The corollary simply follows from the fact that $\pol(H)=\sum_{\omega \in \Omega} \var(X_\omega)$ and the described information structure simultaneously maximizes $\var(X_\omega)$ for all $\omega\in \Omega$.

\section{Vectorial Feasibility}\label{sec:vector}

In this section we connect our results on the feasibility of distributions over \emph{empirical distributions of beliefs} with the notion of the feasibility of distributions over \emph{vectors of beliefs}. As mentioned in the Introduction, the latter notion of feasibility has been previously studied in \citep{dawid1995coherent,arieli2021feasible,privateprivate}. To distinguish between the two notions we refer to the latter notion as \emph{vectorial feasibility}.

The model is similar to the one presented in Section \ref{sec:main}. Given an information structure $I\in \Delta(\Omega \times S_1 \times ... \times S_n)$, every realization of signals $(s_i)_{i\in [n]}$ induces a vector of beliefs $(x_i(s_i))_{i\in [n]}$. Hence, every information structure induces a distribution over vectors of beliefs $\vv_I \in \Delta((\Delta(\Omega)^n)$.

\begin{definition}
	A distribution $\vv \in \Delta((\Delta(\Omega)^n)$ is \emph{vector feasible} if there exists $I$ such that $\vv=\vv_I$. 
\end{definition}

Every distribution $\vv\in \Delta((\Delta(\Omega)^n)$ induces a distribution $\p(\vv)\in \Delta(\Delta_n(\Delta(\Omega)))$ simply because every vector of beliefs induces an empirical distribution of beliefs in the population. 
We observe that the feasibility of $\p(\vv)$ is a necessary but insufficient condition for the feasibility of $\vv$. Indeed this condition is necessary, 
because if there exists no $I$ that will induce the distribution of empirical distributions $\p(\vv)$ then $\vv$ cannot be implemented by any $I$.
It is also not hard to see that the feasibility of $\p(\vv)$ is an insufficient condition for the feasibility of\footnote{To construct a counterexample consider $|\Omega|=2$, $n=2$ and consider any \emph{feasible} vector distribution $\vv'$ with $\vv'(a,b)>0$ for some $a\neq b$. Let $\vv$ be a distribution that is obtained from $\vv'$ by moving some positive mass from $(a,b)$ to $(b,a)$. The distribution $\vv$ becomes vector infeasible because it violates the martingale condition for both agents. However $\p(\vv)=\p(\vv')$ remains feasible.
	
	Another counterexample that does not violate the martingale condition for any player is the following. Let $|\Omega|=2$ and $n=3$. The distribution $\vv=\frac{1}{2}\delta_{(\frac{1}{3},\frac{1}{3},\frac{2}{3})}+\frac{1}{2}\delta_{(\frac{2}{3},\frac{2}{3},\frac{1}{3})}$ is vector infeasible because agents 1 and 2 are agreeing to disagree with agent 3 in both points of the support of this distribution. However, $\p(\vv)$, which is the distribution where with probability $\frac{1}{2}$ two of the three agents have a posterior of $\frac{j}{3}$ (without specifying their identity) and the remaining agent has a posterior of $\frac{3-j}{3}$ for $j=1,2$, is feasible by Corollary \ref{cor:binary}.  
} $\vv$.

Therefore, the tools developed in this paper might be useful for determining the vectorial infeasibility of $\vv$ by checking our condition for the feasibility of  $\p(\vv)$. The following subsection demonstrates that in sufficiently symmetric cases the feasibility of $\p(\vv)$ serves as a necessary \emph{and sufficient} condition for the vectorial feasibility of $\vv$. In particular, using our results we are able to deduce new insights about feasible and infeasible vector distributions $\vv$.

\subsection{Private private information}

A special case of distributions $\vv$ whose vectorial feasibility has been studied in \citep{privateprivate, arieli2021feasible,gutmann1991existence} is that of $|\Omega|=2$ and a \emph{product} distribution $V=\Pi_{i\in [n]} Q_i$ where $Q_i \in \Delta([0,1])$ is the distribution of the beliefs of agent $i$. \citep{privateprivate} refer to this class of distributions as \emph{private private information} because the private information of agents other than $i$ does not affect their beliefs about agent $i$'s belief. In any case, the belief of every agent $j\neq i$ about agent $i$'s belief will be $Q_i$. For a discussion on the attractiveness of this property see \citep{privateprivate,arieli2021feasible}. A special case of private private information that has received attention in the existing literature (see \citep{privateprivate, arieli2021feasible,gutmann1991existence}) is the \emph{symmetric} case where $V=Q^n$. The characterization of the feasible and infeasible $V=Q^n$ is known for the case of $n=2$ (\citep{privateprivate}) but remains unclear for $n\geq 3$.

Consider now the case of a finitely supported $Q$. For the case of $V=Q^n$ the induced distribution over empirical distributions $\p(\vv)$ is the multinomial distribution. More formally, assume that $Q$ is supported on $x_1,...,x_k\in [0,1]$ and assigns a probability of $\alpha_i$ to $x_i$. Let $\mn(n,Q)$ be the multinomial distribution that draws $n$ samples from $\{x_i\}_{i\in [k]}$ according to $Q$. In the standard definition of the multinomial distribution it is assumed that the outcome is a vector of integers $(n_1,...,n_k)$ that sum up to $n$ and indicates the number of times $x_1,...,x_k$ have been chosen. In our case, we refer to the outcome of $\mn(n,Q)$ as $(\frac{n_1}{n},...,\frac{n_k}{n})$ and hence the outcome is an element of $\Delta_n(\Delta(\Omega))$. The following proposition shows that in this case the feasibility of $\p(\vv)$ is a necessary and sufficient condition for the vectorial feasibility of $\vv=Q^n$. In particular, it provides a novel characterization of vectorial feasibility for this class of distributions.     
\begin{proposition}
	For a binary state space $\Omega$ and $Q\in\Delta([0,1])$ the product distribution $\vv=Q^n$ for $n$ agents is vector feasible if and only if the multinomial distribution $\mn(n,Q)$ is a mean-preserving spread of\footnote{We recall that $Q^\omega\in \Delta([0,1])$ denotes the distribution of beliefs conditional on state $\omega$ in the case where the ex-ante distribution of beliefs is $Q$.} $(1-\mu) Q^{\omega=0} + \mu Q^{\omega=1}$.
\end{proposition}

\begin{proof}
	Denote $\p=\p(\vv)=\mn(n,Q)$. Note that $\tilde{\p}=Q$. By Theorem \ref{theorem:main}, $\p$ is feasible if and only if $\p$ is a mean-preserving spread of $(1-\mu) Q^{\omega=0} + \mu Q^{\omega=1}$.
	
	If $\p=\p(\vv)$ is infeasible then $\vv$ is not vector feasible, as discussed above.
	
	If $\p$ is feasible, we observe that $\vv$ can be viewed as the \emph{unique} distribution $\ww$ with the following two properties:
	
	\begin{enumerate}
		\item $\p(\ww)=\p$.
		\item $\ww$ assigns an equal probability to every two vectors $(y_1,...,y_n),(y'_1,...,y'_n)\in \{x_1,...,x_k\}^n$ with the same empirical distributions of beliefs.
	\end{enumerate}
	It is easy to check that our construction of the information structure in the proof of Theorem \ref{theorem:main} induces a distribution over vectors of beliefs that satisfies these two properties and hence this information structure implements $\vv$. Thus, $\vv$ is feasible. 
\end{proof}

Determining whether $\mn(n,Q)$ is a mean-preserving spread of $(1-\mu) Q_0 + \mu Q_1$ might not be an easy task in the case where $k$ (i.e., the support size of $Q$) is large. This is because $\mn(n,Q)$ is a distribution over elements in $\Delta([k])$, and in high dimensions the mean-preserving condition might not be easy to check. However, for $k=2$ the multinomial distribution becomes simply a binomial distribution and the mean-preserving condition is easily verifiable by Lemma \ref{lem:mps}. Thus, for $Q$ with a binary support we get the following condition for the vectorial feasibility of $\vv=Q^n$.

\begin{corollary}\label{cor:bin-private}
	Let $Q\in \Delta(\{a,b\})$ where $0<a<\mu<b<1$. The distribution $\vv=Q^n$ is vector feasible (for $n$ agents) if and only if\footnote{We adopt here our previous convention for the multinomial distribution, and assume that the Binomial distribution receives values in $\{0,\frac{1}{n},...,1\}$ and not in $\{0,1,...,n\}$.} $\mathbb{E}[\bin(n, \frac{\mu-a}{b-a})|_{\mu}]\leq \frac{(\mu-a)b}{(b-a)\mu}$.
\end{corollary}
\noindent
That is, the corollary states that in order to determine the feasibility of $\vv=Q^n$ one should calculate the expectation of the $\mu$-quantile of the Binomial distribution. If it exceeds some constant, then the distribution is infeasible. Otherwise it is feasible. Below is a case where we can provide an explicit formula for the feasibility for every number of agents.

\begin{example}
	\begin{figure}[h]
		\centering
		\begin{tikzpicture}
			\draw[->] (1,0)--(12,0);
			\node[right] at (12,0) {$n$};
			
			\draw[->] (1,0)--(1,6);
			\node[above] at (1,6) {$a$};
			
			\foreach \i in {2,...,11} {
				\draw (\i,-0.1)--(\i,0.1);
				\node[below] at (\i,0) {\i};
			};
			
			\draw (0.9,5)--(1.1,5);
			\node[left] at (1,5) {0.5};
			
			\draw (0.9,2.5)--(1.1,2.5);
			\node[left] at (1,2.5) {0.25};
			
			\draw (0.9,4)--(1.1,4);
			\node[left] at (1,4) {0.4};
			
			\draw[line width=2mm] (2,2.5)--(2,5);
			
			\draw[line width=2mm] (3,2.5)--(3,5);
			
			\draw[line width=2mm] (4,3.125)--(4,5);
			
			\draw[line width=2mm] (5,3.125)--(5,5);
			
			\draw[line width=2mm] (6,3.44)--(6,5);
			
			\draw[line width=2mm] (7,3.44)--(7,5);
			
			\draw[line width=2mm] (8,3.63)--(8,5);
			
			\draw[line width=2mm] (9,3.63)--(9,5);
			
			\draw[line width=2mm] (10,3.77)--(10,5);
			
			\draw[line width=2mm] (11,3.77)--(11,5);
			
			%\draw (0.9,3)--(1.1,3);
			%\node[left] at (1,3) {0.3};
		\end{tikzpicture}
		\caption{The values of $a$ for which the product distribution $Q^n$ is feasible as a function of the number of agents $n$.}
		\label{fig:label}
	\end{figure}
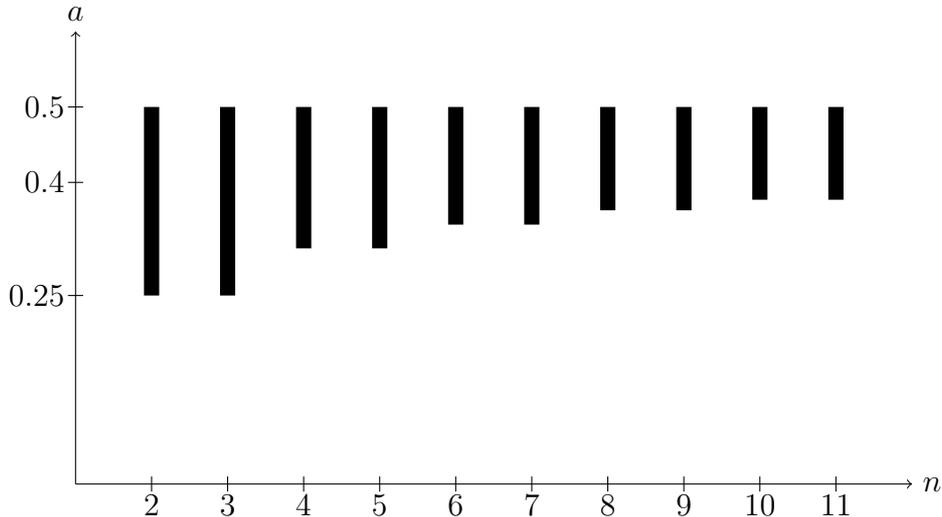
	
	Let $\mu=\frac{1}{2}$ and let $Q=\frac{1}{2}\delta_a+\frac{1}{2}\delta_{1-a}$ for $0<a<\frac{1}{2}$. For $n=2m$ and for $n=2m+1$, the distribution $\vv=Q^n$ is feasible if and only if $a\geq \frac{1}{2} - \binom{2m}{m} 2^{-2m-1}$ (see Figure \ref{fig:label}). This observation is a consequence of Corollary \ref{cor:bin-private}, where the expectation of the $\frac{1}{2}$-quantile of the $\bin(n,\frac{1}{2})$ distribution can be explicitly computed and is provided in the following lemma.
	
	\begin{lemma}\label{lem:quantile}
		For every $m\in \mathbb{N}$ we have $$\mathbb{E}[\bin(2m,p)|_{\frac{1}{2}}]=\mathbb{E}[\bin(2m+1,p)|_{\frac{1}{2}}]=\frac{1}{2} - \binom{2m}{m} 2^{-2m-1}.$$
	\end{lemma}
	The proof of this lemma is relegated to Appendix \ref{sec:bin}.

	Asymptotically, for large population of $n$ agents, the threshold value of $a$ for which $V=Q^n$ is feasible is $a\geq \frac{1}{2} - \binom{2m}{m} 2^{-2m-1}\approx \frac{1}{2}- \frac{1}{\sqrt{2\pi n}}$. Another interesting phenomenon is that the threshold value of $a$ is identical for an even $n$ and for $n+1$.  
\end{example}

%\section{Applications}

%In this section we discuss application for the case of a binary state space $\Omega=\{0,1\}$.

\section{Private Bayesian Persuasion of a Homogeneous Population}\label{sec:pbp}

We next consider an application of our result to private Bayesian persuasion. As in \citep{arieli2019private}, consider a sender who is a marketer that tries to sell a product to a group of agents.  Consider a binary state space $\Omega=\{0,1\}$ with a common prior $\mu$.  The agents' action set is $A=\{0,1\}$, where action $a=1$ represents a decision to buy the product and action $a=0$ represents a decision not to buy the product. Playing $a=0$ yields a utility of $0$. The state represents the quality of the product. In $\omega=1$ the product is of good quality and worth $1$ to the agents while in state $\omega=0$ it is worth $-\frac{\tau}{1-\tau}$ to the agent, where $\tau\in(0,1)$. Thus, at a posterior of $\tau$ the agents are indifferent between buying and not buying the product. We assume that the agents share no payoff externalities. Let $u:[0,1]\rightarrow\real$ be the utility of the sender, where $u(x)$ represents her utility from persuading a fraction $x$ of agents to take action 1. We assume that $u$ is non-decreasing and that $\mu<\tau$ (otherwise the policy that reveals no information would be optimal).
We denote by $V_n$ the optimal utility the sender can achieve in an interaction with $n$ agents.

\citep{arieli2019private,barman} study a heterogeneous variant of this setting, where each agent possesses an individual threshold $\tau_i \in (0,1)$. In the heterogeneous case a clean formula for $V_n$ has been deduced only in the case where $u$ is concave; see \citep{arieli2019private}.\footnote{A polynomial algorithm for computing $V_n$ has been deduced in \citep{barman}.} Here we observe that the homogeneous population variant of the problem can be solved using our main result (Theorem \ref{theorem:main}) and a clean formula for $V_n$ can be provided.    

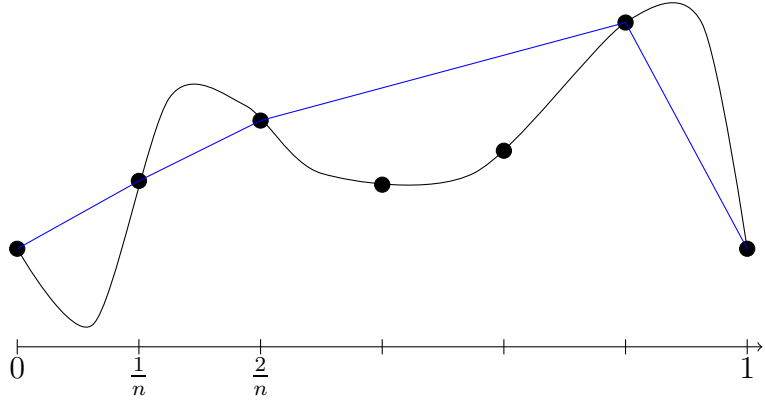
\begin{figure}[h]
	\centering
	\begin{tikzpicture}
		\draw[black] plot [smooth] coordinates {(0,1) (1,0) (2,3) (3,2.9) (4,2) (6,2) (8,4) (9,4) (9.6,1)};
		
		\filldraw (0,1) circle (0.1);
		
		\filldraw (1.6,1.9) circle (0.1);
		
		\filldraw (3.2,2.7) circle (0.1);
		
		\filldraw (4.8,1.85) circle (0.1);
		
		\filldraw (6.4,2.3) circle (0.1);
		
		\filldraw (8,4) circle (0.1);
		
		\filldraw (9.6,1) circle (0.1);
		
		\draw[->] (0,-0.3) -- (9.8,-0.3);
		
		\draw (0,-0.2)--(0,-0.4);
		\draw (1.6,-0.2)--(1.6,-0.4);
		\draw (3.2,-0.2)--(3.2,-0.4);
		\draw (4.8,-0.2)--(4.8,-0.4);
		\draw (6.4,-0.2)--(6.4,-0.4);
		\draw (8,-0.2)--(8,-0.4);
		\draw (9.6,-0.2)--(9.6,-0.4);
		
		\node[below] at (0,-0.3) {0};
		\node[below] at (1.6,-0.3) {$\frac{1}{n}$};
		\node[below] at (3.2,-0.3) {$\frac{2}{n}$};
		\node[below] at (9.6,-0.3) {1};
		
		\draw[blue] (0,1) -- (1.6,1.9) -- (3.2,2.7) --  (8,4) -- (9.6,1);

	\end{tikzpicture}
	\caption{Concavification over a grid. The function $u$ is displayed in black. The concavification over the grid  $\cav_n (u)$ is displayed in blue.}
	\label{fig:my_label}
\end{figure}

We denote by $\cav_{n}(u):[0,1]\to \real$ the concavification of $u$ when $u$ is restricted to the $\frac{1}{n}$-grid points $\{0,\frac{1}{n},\frac{2}{n},...,1\}$ only (see Figure \ref{fig:my_label}). Namely, $\cav_{n}(u)$ is the minimial concave function that satisfies $\cav_{n}(u)(x)\geq u(x)$ for every $x\in \{0,\frac{1}{n},\frac{2}{n},...,1\}$. Equivalently, for every $y\in [0,1]$ we define
\begin{align}\label{eq:conc}
	\cav_{n}(u)(y)=\max_{Q\in \Delta_n(\{0,1\}), \mathbb{E}[Q]=y} \mathbb{E}_{x\sim Q} [u(x)].   
\end{align}
The characterization of the sender's optimal value is as follows.

\begin{proposition}\label{pro:pbp}
	The optimal utility for the sender is
	$$V_n=\mu u(1)+(1-\mu)\cav_n(u)(\frac{\mu(1-\tau)}{\tau(1-\mu)}).$$
\end{proposition}
The formula for the value can be interpreted as follows. In state $\omega=1$ all agents take action 1 and the utility is $u(1)$. In state $\omega=0$ an expected fraction of $\frac{\mu(1-\tau)}{\tau(1-\mu)}$ agents adopt the product and the sender is allowed to optimize over all the distributions with expectation $\frac{\mu(1-\tau)}{\tau(1-\mu)}$ in order to achieve $\cav_n(u)(\frac{\mu(1-\tau)}{\tau(1-\mu)})$.

Note that the concavification in Proposition \ref{pro:pbp} differs from the classic characterization of the value via concavification (\citep{kamenica2011bayesian}). The standard characterization applies concavification to a function whose domain is that of the posterior beliefs. Our characterization (for this special case) applies concavification to a function whose domain is that of a fraction of the buyers.

Proposition \ref{pro:pbp} also provides an asymptotic characterization of the value for large populations. Denote $V^*=\lim_{n \to \infty} V_n$ and denote by $\cav(u)$ the standard notion of concavification (without restricting it to a grid). By Proposition \ref{pro:pbp} we have $$V^*=\mu u(1)+(1-\mu)\cav(u)(\frac{\mu(1-\tau)}{\tau(1-\mu)})$$
because $\cav_n(u)$ point-wise converges to $\cav(u)$ as $n\to \infty$ for non-decreasing functions $u$. 

\begin{proof}[Proof of Proposition \ref{pro:pbp}]
	The idea is to utilize Corollary \ref{cor:binary}. To do so, we should be able to identify the possible posteriors that agents could have in an optimal policy, and to prove that there are only two possible values for these posteriors. To this end, we utilize the following observation from \citep{arieli2019private}.
	
	\begin{lemma}[\citep{arieli2019private}]\label{lem:jet}
		There exists an optimal policy with the following properties:
		
		\begin{itemize}
			\item The policy uses binary signals $S_i=\{s^i_0,s^i_1\}$ for each agent $i$, and each agent takes action $1$ if and only if she receives the signal $s^i_1$.
			\item In state $\omega=1$ signal $s^i_1$ is sent to all agents with probability $1$.
			\item In state $\omega=0$ each agent $i$ gets the signal $s^i_1$ with probability\footnote{The Lemma does not indicate on the correlation of the $s^i_1$ signals among the agents conditional on state $\omega=0$.}  $\frac{\mu(1-\tau)}{\tau(1-\mu)}$.
		\end{itemize}
	\end{lemma}
	
	By Lemma \ref{lem:jet} we deduce that the only two possible posteriors for an agent are $0$ and $\tau$. 
	Now we can use Corollary \ref{cor:binary}, which states that $\p\in\Delta(\{0,\frac{1}{n},\frac{2}{n},...,1\})$ is feasible if and only if $\p$ is a mean-preserving spread of  $\mu\delta_{1}+(1-\mu)\delta_{\frac{\mu(1-\tau)}{\tau(1-\mu)}}$, where $\p$ is the number of agents with posterior $\tau$. This set of mean-preserving spreads consists of all $\mu \delta_1 +(1-\mu) Q$ where $Q\in \Delta(\{0,\frac{1}{n},\frac{2}{n},...,1\})$ satisfy $\mathbb{E}[Q]=\frac{\mu(1-\tau)}{\tau(1-\mu)}$ and the utility of the sender for each such mean-preserving spread is $\mu u(1)+(1-\mu) \mathbb{E}_{x\sim Q} [u(x)]$, which is precisely the concavification formula of Equation \eqref{eq:conc}. 
\end{proof}
\bibliographystyle{plainnat}
\bibliography{bibfile}

\begin{thebibliography}{21}
\providecommand{\natexlab}[1]{#1}
\providecommand{\url}[1]{\texttt{#1}}
\expandafter\ifx\csname urlstyle\endcsname\relax
  \providecommand{\doi}[1]{doi: #1}\else
  \providecommand{\doi}{doi: \begingroup \urlstyle{rm}\Url}\fi

\bibitem[Arieli and Babichenko(2019)]{arieli2019private}
Itai Arieli and Yakov Babichenko.
\newblock Private {B}ayesian persuasion.
\newblock \emph{Journal of Economic Theory}, 182:\penalty0 185--217, 2019.

\bibitem[Arieli et~al.(2020)Arieli, Babichenko, and
  Smorodinsky]{arieli2020identifiable}
Itai Arieli, Yakov Babichenko, and Rann Smorodinsky.
\newblock Identifiable information structures.
\newblock \emph{Games and Economic Behavior}, 120:\penalty0 16--27, 2020.

\bibitem[Arieli et~al.(2021)Arieli, Babichenko, Sandomirskiy, and
  Tamuz]{arieli2021feasible}
Itai Arieli, Yakov Babichenko, Fedor Sandomirskiy, and Omer Tamuz.
\newblock Feasible joint posterior beliefs.
\newblock \emph{Journal of Political Economy}, 129\penalty0 (9):\penalty0
  2546--2594, 2021.

\bibitem[Aumann(1976)]{aumann1976agreeing}
Robert~J Aumann.
\newblock Agreeing to disagree.
\newblock \emph{Annals of Statistics}, 4\penalty0 (6):\penalty0 1236--1239,
  1976.

\bibitem[Aumann et~al.(1995)Aumann, Maschler, and Stearns]{aumann1995repeated}
Robert~J Aumann, Michael Maschler, and Richard~E Stearns.
\newblock \emph{Repeated {G}ames with {I}ncomplete {I}nformation}.
\newblock MIT Press, 1995.

\bibitem[Babichenko and Barman(2016)]{barman}
Yakov Babichenko and Siddharth Barman.
\newblock Computational aspects of private {B}ayesian persuasion.
\newblock \emph{arXiv preprint arXiv:1603.01444}, 2016.

\bibitem[Blackwell(1951)]{blackwell1951comparison}
David Blackwell.
\newblock Comparison of experiments.
\newblock In \emph{Proceedings of the Second Berkeley Symposium on Mathematical
  Statistics and Probability}, pages 93--102. University of California Press,
  1951.

\bibitem[Burdzy and Pal(2019)]{burdzy2019contradictory}
Krzysztof Burdzy and Soumik Pal.
\newblock Contradictory predictions.
\newblock \emph{arXiv preprint arXiv:1912.00126}, 2019.

\bibitem[Burdzy and Pitman(2020)]{burdzy2020bounds}
Krzysztof Burdzy and Jim Pitman.
\newblock Bounds on the probability of radically different opinions.
\newblock \emph{Electronic Communications in Probability}, 25:\penalty0 1--12,
  2020.

\bibitem[Dawid et~al.(1995)Dawid, DeGroot, Mortera, Cooke, French, Genest,
  Schervish, Lindley, McConway, and Winkler]{dawid1995coherent}
AP~Dawid, MH~DeGroot, J~Mortera, R~Cooke, S~French, C~Genest, MJ~Schervish,
  DV~Lindley, KJ~McConway, and RL~Winkler.
\newblock Coherent combination of experts' opinions.
\newblock \emph{Test}, 4\penalty0 (2):\penalty0 263--313, 1995.

\bibitem[De~Oliveira and Lamba(2019)]{de2019rationalizing}
Henrique De~Oliveira and Rohit Lamba.
\newblock Rationalizing dynamic choices.
\newblock \emph{Available at SSRN 3332092}, 2019.

\bibitem[Dubins and Pitman(1980)]{dubins1980maximal}
Lester~E Dubins and Jim Pitman.
\newblock A maximal inequality for skew fields.
\newblock \emph{Zeitschrift f{\"u}r Wahrscheinlichkeitstheorie und verwandte
  Gebiete}, 52\penalty0 (3):\penalty0 219--227, 1980.

\bibitem[Dughmi and Xu(2019)]{dughmi2019algorithmic}
Shaddin Dughmi and Haifeng Xu.
\newblock Algorithmic {B}ayesian persuasion.
\newblock \emph{SIAM Journal on Computing}, \penalty0 (3):\penalty0 STOC16--68,
  2019.

\bibitem[Gutmann et~al.(1991{\natexlab{a}})Gutmann, Kemperman, Reeds, and
  Shepp]{gutmann}
Sam Gutmann, JHB Kemperman, JA~Reeds, and Larry~A Shepp.
\newblock Existence of probability measures with given marginals.
\newblock \emph{The Annals of Probability}, 19\penalty0 (4):\penalty0
  1781--1797, 1991{\natexlab{a}}.

\bibitem[Gutmann et~al.(1991{\natexlab{b}})Gutmann, Kemperman, Reeds, and
  Shepp]{gutmann1991existence}
Sam Gutmann, JHB Kemperman, JA~Reeds, and Larry~A Shepp.
\newblock Existence of probability measures with given marginals.
\newblock \emph{The Annals of Probability}, 19\penalty0 (4):\penalty0
  1781--1797, 1991{\natexlab{b}}.

\bibitem[He et~al.(2021)He, Sandomirskiy, and Tamuz]{privateprivate}
Kevin He, Fedor Sandomirskiy, and Omer Tamuz.
\newblock Private private information.
\newblock \emph{arXiv preprint arXiv:2112.14356}, 2021.

\bibitem[Kamenica and Gentzkow(2011)]{kamenica2011bayesian}
Emir Kamenica and Matthew Gentzkow.
\newblock Bayesian persuasion.
\newblock \emph{American Economic Review}, 101\penalty0 (6):\penalty0
  2590--2615, 2011.

\bibitem[Milgrom and Stokey(1982)]{milgrom1982information}
Paul Milgrom and Nancy Stokey.
\newblock Information, trade and common knowledge.
\newblock \emph{Journal of Economic Theory}, 26\penalty0 (1):\penalty0 17--27,
  1982.

\bibitem[Molavi(2021)]{molavi2021tests}
Pooya Molavi.
\newblock Tests of {B}ayesian rationality.
\newblock \emph{arXiv preprint arXiv:2109.07007}, 2021.

\bibitem[Morris(2020)]{morris2020notrade}
Stephen~E Morris.
\newblock No trade and feasible joint posterior beliefs.
\newblock \emph{Working paper}, 2020.

\bibitem[Shmaya and Yariv(2016)]{shmaya2016experiments}
Eran Shmaya and Leeat Yariv.
\newblock Experiments on decisions under uncertainty: A theoretical framework.
\newblock \emph{American Economic Review}, 106\penalty0 (7):\penalty0
  1775--1801, 2016.

\end{thebibliography}

\appendix
\section{Extending the Proof of Theorem \ref{theorem:main} to infinitely supported distributions}\label{app:trm-proof}

We have proved in Section \ref{sec:proof} that a  finitely supported feasible $\p$ must be a mean-preserving spread of $\p'$. Here we extend this result to arbitrary feasible distributions $\p$. 

Let $G$ be an information structure $G$ for $n$ agents where $\{S_i\}_{i=1}^n$ are (possibly infinite) measurable sets. Note that the measure $\p_G$ that is induced by $G$ can be weakly approximated using $\p_{G'}$ when $G'$ is a finite information structure. 
To do this, we use the following standard trick. First we partition $\Delta(\Omega)$ using a finite partition $\mathcal{D}_\delta$ such that each element $D\in\mathcal{D}_\delta$ is a convex set and has a diameter at most $\delta$.
We define the new signal space $S'_i=\mathcal{D}_\delta$ for every agent $i$.
We then couple together $G$ and $G'$ as follows. We draw a signal profile according to $G$ and tell each agent $i$ only of the partition element in $\mathcal{D}_\delta$ in which the posterior $\p(\omega=1|s_i)$ lies.
It is easy to see that $\p_{G'}$ weakly converges to $\p$ as $\delta$ goes to $0$. 
%and hence by definition $\p\in\mathcal{P}$.

It now remains to show that if $\{\p_k\}_k\subset\Delta(\Delta_n(\Delta(\Omega)))$ is a sequence of measures with finite support that weakly converges to $\p$ such that $\p_k$ %=\p_{G_{n_k}}$ 
is a mean-preserving spread of $\p'$ for every $k$, then $\p$ also satisfies the this condition.

%This readily follows from the fact the conditions stated in Theorem \ref{theorem:main} are weakly closed conditions. 
Note that $\p_k=\sum_{\omega \in \Omega}\mu_\omega \p^\omega_{k}$. Since $\Delta(\Delta_n(\Delta(\Omega))) $ is a weakly compact space we can assume (by taking subsequences if necessary) that all sequences $\{\p^\omega_{k}\}_k$ weakly converge to a measure $\p^\omega$ for every state $\omega\in\Omega$. Therefore, 
$\p=\sum_{\omega \in \Omega} \mu_\omega \p_{\omega}$. We further note that taking the expected measure is a weakly continuous operator from $\Delta(\Delta_n(\Delta(\Omega)))$ to $\Delta(\Delta(\Omega))$. Therefore, the sequences $\{\tp_k\}_k$ and $\{\tp^\omega_{k}\}_k$ converge to $\tp$ and $\tp^\omega$ respectively for every $\omega$. Therefore $\p$ is a mean-preserving spread of $\p'=\sum_{\omega \in \Omega} \mu_\omega \delta_{\tp^\omega}$ as desired.

Furthermore, since for every $k$ and $\omega\in\Omega$,  $$\mathrm d\tp^\omega_{k}(x)=\frac{1}{\mu_\omega}  x_\omega \mathrm d\tp_k(x)
,$$
it follows that in the limit for every $\omega\in\Omega$,
$$\mathrm d\tp^\omega(x)=\frac{1}{\mu_\omega}  x_\omega\mathrm d\tp(x),$$
as desired.
%To following lemma completes the proof of the necessity direction.

%\begin{lemma}
%If $\{\p_n\}_{n}^\infty\subset\Delta(\Delta([0,1]))$ is a sequence of measures for which $\p_n$ is a mean preserving spread of $\p'_n$ for every $n$ and $\lim_{n\gor\infty}\p_n=\p$, then $\p$ is a mean preserving spread of $\p'$.
%\end{lemma}

\section{Proof of Lemma \ref{lem:quantile}}\label{sec:bin}
The distribution $\bin(2m+1,\frac{1}{2})|_{\frac{1}{2}}$ assigns a probability of $\binom{2m+1}{k}2^{-2m}$ to $\frac{k}{2m+1}$ for every $k\leq m$, hence its expectation is given by
\begin{align}\label{eq:bin}
\begin{aligned}
	\mathbb{E}[\bin(2m+1,\frac{1}{2})|_{\frac{1}{2}}] &=\sum_{k=1}^m \binom{2m+1}{k}2^{-2m} \frac{k}{2m+1} =\frac{2^{-2m}}{2m+1}\sum_{k=1}^m \binom{2m+1}{k}k \\
	&= \frac{2^{-2m}}{2m+1} (2m+1)\sum_{k=0}^{m-1} \binom{2m}{k}=2^{-2m} \frac{1}{2} \left(2^{2m}- \binom{2m}{m} \right) \\
	&= \frac{1}{2} - \binom{2m}{m} 2^{-2m-1}
\end{aligned}
\end{align}
where the third equality follows from the fact that both expressions $\sum_{k=1}^m \binom{2m+1}{k}k$ and $(2m+1)\sum_{k=0}^{m-1} \binom{2m}{k}$ count the number of possible choices of a subset of at most $k$ people and their leader out of a group of $2m+1$ people. The first expression first chooses the group and thereafter the leader. The second expression first chooses the leader and thereafter complements it with a subset.

The distribution $\bin(2m,\frac{1}{2})|_{\frac{1}{2}}$ assigns a probability of $\binom{2m}{k}2^{-2m+1}$ to $\frac{k}{2m}$ for every $k<m$ and assigns a probability of $\binom{2m}{m}2^{-2m}$ to $\frac{1}{2}$, and hence its expectation is given by
\begin{align*}
\mathbb{E}[\bin(2m,\frac{1}{2})|_{\frac{1}{2}}] &=\sum_{k=1}^{m-1} \binom{2m}{k}2^{-2m+1} \frac{k}{2m} + \binom{2m}{m}2^{-2m-1} \\
& =\frac{2^{-2m+1}}{2m} \sum_{k=1}^{m-1} \binom{2m}{k} k   + \binom{2m}{m}2^{-2m-1} \\
&=\frac{2^{-2m+1}}{2m} 2m \sum_{k=0}^{m-2} \binom{2m-1}{k}   + \binom{2m}{m}2^{-2m-1} \\
&=2^{-2m+1} \frac{1}{2} \left( 2^{2m-1} - 2 \binom{2m-1}{m-1} \right)   + \binom{2m}{m}2^{-2m-1} \\
&=\frac{1}{2} - 2^{-2m+1}  \binom{2m-1}{m-1}  + \binom{2m}{m}2^{-2m-1} \\
&= \frac{1}{2} - \binom{2m}{m} 2^{-2m-1}
\end{align*}
where the argument for the third equality is similar to the third equality in Equation \eqref{eq:bin}.
\end{document}